\documentclass[10pt]{article} 
\usepackage{fullpage}
\usepackage[latin1]{inputenc} 

\usepackage{indentfirst}
\usepackage{authblk}

\usepackage{dustins_essential_macros}

\newcommandx*{\FT}[2][1=h,2=k]{\text{FT}^{#1}(#2)}
\newcommandx*{\BT}[2][1=h,2=k]{\text{BT}^{#1}(#2)}

\newcommandx*{\BDE}[2][1=G,2=k]{\ensuremath{\text{DE}_{#1}(#2)}\xspace}
   
\newcommandx*{\BDEnok}[1][1=G]{\ensuremath{\text{DE}_{#1}}\xspace}
\newcommandx*{\bde}{\ensuremath{\text{DE}}\xspace}
\newcommandx*{\bdeG}[1][1=G]{\ensuremath{\text{DE}_{#1}}\xspace}

\def\av{\vec{a}}

\def\root{u_{\text{root}}}
\def\m{m}

\newcommandx*{\GEN}[1][1] {
 \def\NAMEGEN{\text{GEN}}
 \ifthenelse{ \isempty{#1} }
             { \NAMEGEN}   
             { #1\text{-}\NAMEGEN }
}

\newcommandx*{\val}[2][1=u,2=I]{#1^{#2}}
\newcommandx*{\fval}[3][3=I]{f^{#3}_{#1}(#2)}
\newcommandx*{\lval}[2][2=I]{l^{#2}_{#1}}

\newcommand{\nhasv}[2]{\ensuremath{\pair{#1,#2}}} 
\newcommand{\ehasv}[2]{\ensuremath{\pair{#1,#2}}}
\newcommand{\initial}{1}
\newcommand{\indexofnode}[1]{\ensuremath{#1+1}}

\newcommand{\Outputs}{\mathit{Out}}
\newcommand{\Vars}{\mathit{Vars}} 
\newcommand{\asciitilde}{{\raise.17ex\hbox{$\scriptstyle\sim$}}}

\title{Lower bound for deterministic semantic-incremental branching programs solving $\GEN$}
\author{Dustin Wehr\footnote{www.cs.toronto.edu/{\asciitilde}wehr}} 
\affil{University of Toronto}
\begin{document}
\maketitle
 
\begin{abstract}
We answer a problem posed in \cite{McKenzie} regarding a restricted model of small-space computation, tailored for solving the $\GEN$ problem. They define two variants of ``incremental branching programs'', the \emph{syntactic} variant defined by a restriction on the graph-theoretic paths in the program, and the more-general \emph{semantic} variant in which the same restriction is enforced only on the consistent paths - those that are followed by at least one input. 
They show that exponential size is required for the syntactic variant, but leave open the problem of superpolynomial lower bounds for the semantic variant. 
Here we give an exponential lower bound for the semantic variant by generalizing lower bound arguments from \cite{FSTTCS} \cite{myMasters} for a similar restricted model tailored for solving a special case of $\GEN$ called Tree Evaluation.
\end{abstract}

\tableofcontents

%

\section{How to read this paper}
The introduction \ref{s:intro} (which is short and should be read entirely) defers several definitions to Section \ref{s:defs}; the reader should refer there to read any unfamiliar definitions as they arise. 
The proof of our main result spans Sections \ref{s:outline}, \ref{s:thrifty_lb}, \ref{s:sem_incr_lb}. All but the most-casual readers should read Section \ref{s:outline}, which sets up and outlines the proof. Sections \ref{s:remarksMcKenzie} and \ref{s:easy_inputs} can safely be skipped by readers only interested in our main result.

\section{Introduction} \label{s:intro} 

An instance $T$ of $\GEN[\m]$ is just a function from $[\m] \times [\m]$ to $[\m]$, where $[\m] = \{1,\ldots,\m\}$. $T$ is a YES instance iff $\m$ is in the closure of the set $\{1\}$ under the operation $T$. Depending on the computation model, $T$ is given as $\m^2 \lceil \log \m \rceil$ bits or, more naturally, as $\m^2$ elements of $[\m]$. The computation model we use, $\m$-way branching programs (defn \ref{d:k-way_BPs}), is the standard one for studying the nonuniform space complexity of problems represented in the second way. 

We refer to the $\m^2$, $[\m]$-valued input variables that define the $\GEN[\m]$ instances by the names $\{(x,y)\}_{x,y \in [\m]}$, and we refer to input $T$'s value of variable $(x,y)$ by $T(x,y)$.
Throughout, we only talk about deterministic branching programs. A {\bf deterministic semantic-incremental} branching program (BP) solving $\GEN[\m]$ is an $\m$-way BP $B$ that computes $\GEN[\m]$ such that for every state $q$ of $B$ that queries a variable $(x,y)$ and every input $T$ that visits $q$, for both $z \in \{x,y\}$ either
$z = 1$ or there is an earlier edge on the computation path of $T$ labeled $z$.
The main goal is Corollary \ref{c:strong_dag_lb} in Section \ref{s:sem_incr_lb}:
\begin{quote}
 There is a constant $c>0$ such that for infinitely-many $\m$ every deterministic semantic-incremental BP solving  $\GEN[\m]$ has at least $2^{c\, \m / \log \m}$ states.
\end{quote}


\section{Preliminaries / Outline} \label{s:prelims}

\subsection{Definitions} \label{s:defs}
\begin{defn}[$k$-way branching program]  \label{d:k-way_BPs}
A \emph{deterministic $k$-way branching program} $B$ computing a function $g:[k]^{|\Vars|} \rightarrow \Outputs$, where $\Vars$ and $\Outputs$ are finite sets, is first of all a directed multi-graph whose nodes are called {\em states}, having a unique in-degree 0 state called the \emph{start state}. 
Every state is labeled with an \emph{input variable} (an element of $\Vars$) except for $|\Outputs|$ {\em output} states with out-degree 0 labelled with distinct \emph{output values} (the elements of $\Outputs$). 
Every state has $k$ out-edges, labeled with distinct elements of $[k]$. 
An \emph{input} $I$ (a mapping $X \mapsto X^I$ from $\Vars$ to $[k]$) defines a \emph{computation path} from the start state through $B$ in the obvious way:  from a non-output state $q$ labeled with $X \in \Vars$,\, $I$ follows the edge out of $q$ labeled $X^I$. The computation path of $I$ must be finite, ending at the output state labeled with $g(I)$. The \emph{size} of $B$ is its number of states. 
We say that $B$ solves a decision problem if $|\Outputs| = 2$.
\end{defn}

\begin{defn}[rooted dag, root $\root$, leaf, child, parent, arc, size] \label{d:rooted_dag}
A {\bf rooted dag} $G$ is a directed acyclic graph with a unique out-degree 0 node called the {\bf root}, denoted $\root$. In-degree 0 nodes are called {\bf leaves}. We refer to the edges of $G$ as {\bf arcs} in order to avoid confusion with the edges of a branching program. The nodes with arcs into $u$ are the {\bf children} of $u$ and the nodes that receive arcs coming out of $u$ are the {\bf parents} of $u$. 
\end{defn}

\begin{defn}[Dag Evaluation Problem] 
An input is a 4-tuple $\langle G,k,\vec{l},\vec{f} \rangle$. $G$ is a connected rooted dag and $k \ge 2$ is an integer. $\vec{l}$ consists of a $\lceil \log k \rceil$-bit string specifying a value in $[k]$ for each leaf node of $G$, and $\vec{f}$ consists of a $k^d \lceil \log k \rceil$-bit string specifying a function from $[k]^d$ to $[k]$ for each non-leaf node of $G$ with $d$ children. Each non-leaf receives a value in the expected way; namely, by applying its function to the values of its children. The function version of the problem asks for the value of the root. The decision version asks if the root value is 1.
\end{defn}

The next definition subsumes the previous one; it makes precise how inputs to the Dag Evaluation Problem are presented to $k$-way BPs, and introduces notation that we will use throughout this paper. The variable $G$ denotes a connected rooted dag (see definition \ref{d:rooted_dag}) with at least two nodes throughout this paper. 
\begin{defn}[\BDEnok \ : Dag Evaluation Problem for fixed dag $G$]
The size of an input to \BDEnok is determined by a parameter $k \ge 2$, and we write \BDE for the problem restricted to inputs with size parameter $k$. The $[k]$-valued input variables $\Vars$ of $\BDE$ are as-follows: 
\begin{equation*} 
\begin{split}
  l_u  & \quad \text{for each leaf $u \in G$} \\
  f_u(\av)  & \quad \text{for each node $u \in G$ of in-degree $d \ge 1$ and each $\av \in [k]^d$} 
\end{split}
\end{equation*}
We write $\lval{u}[I]$ and $\fval{u}{\av}[I]$ for the value input $I {:} \Vars \to [k]$ assigns to variables $l_u$ and $f_u(\av)$. For $u \in G$ we define $\val[u][I]$, the \emph{value of $I$ on $u$}, inductively: if $u$ is a leaf then $\val[u][I] = \lval{u}[I]$, and if $u$ has children $v_1,\ldots,v_d$ then $\val[u][I] = \fval{u}{{\val[v_1]},\ldots,{\val[v_d]}}[I]$. 
$\BDE$ is a decision problem; the output is YES if $\val[\root][I] = 1$ and NO otherwise. 
\end{defn}

We generalize the definition from \cite{FSTTCS} \cite{myMasters} of deterministic thrifty BPs solving the Tree Evaluation Problem (which is the Dag Evaluation Problem for the complete binary trees):
\begin{defn}
 A $k$-way BP solving $\BDE$ is thrifty if for every state $q$ that queries an internal node variable $f_u(a_1,\ldots,a_d)$, if $v_1,\ldots,v_d$ are the children of $u$ then every input $I$ that visits $q$ has $\val[v_1] = a_1, \ldots, \val[v_d] = a_d$. 
 \end{defn}

\begin{defn}[Black pebbling cost of $G$]
Let $G$ be a rooted dag. A {\bf pebbling configuration} $C$ of $G$ is given by a subset of the nodes of $G$ which are said to be {\bf pebbled}. 
A {\bf complete pebbling sequence} for $G$ is a sequence of pebbling configurations $\pi = C_1,\ldots,C_{t^*}$ such that every node is unpebbled in $C_1$, the root is pebbled in $C_t$, and for all $t \in \{1,\ldots,t^*-1\}$, configuration $C_{t+1}$ is obtained from $C_{t}$ by one of the following types of {\bf pebbling moves}:
\begin{ppe}
\item If all the children of node $u$ are pebble in $C_{t}$, then in $C_{t+1}$ a pebble can placed on $u$ and simultaneously zero or more of the children of $u$ can have their pebbles removed. 
\item A pebble is removed from some node.\footnote{We don't actually need to include this as a possible move.}
\end{ppe}  
We say $\pi$ requires $p$ pebbles if $p$ is the maximum over all $C_{t}$ of the number of nodes pebbled in $C_t$. Finally, the {\bf pebbling cost} of $G$ is the minimum number of pebbles required for a complete pebbling sequence for $G$.
\end{defn}

\subsection{Outline of proof} \label{s:outline}
For arbitrary $G$ and $k$, Theorem \ref{t:thrifty_lb} gives lower bounds for thrifty BPs solving $\BDE$ in terms of $k$ and the pebbling cost of $G$. Let $T^h$ be the complete binary tree with $2^h-1$ nodes. Theorem \ref{t:thrifty_lb} is a generalization of the following result from \cite{myMasters}, stated in terms of the notation introduced above:
\begin{quote}
 For any $h,k \ge 2$ every deterministic thrifty BP solving $\BDE[T^h]$ has at least  $k^h$ states. 
\end{quote}
Theorem \ref{t:semantic_dag_main} uses Theorem \ref{t:thrifty_lb} to get lower bounds for semantic-incremental BPs solving $\GEN[\m]$ in terms of the pebbling cost of dags with indegree 2.\footnote{This could be generalized to work for families of dags with unbounded indegree, but we have no use for that generalization here.} The bulk of that proof consists of showing that for any $G$ with indegree 2, there is a polynomial-bounded\footnote{And very efficiently computable, though we don't need that fact.} reduction $g$ from $\BDEnok$ to $\GEN$ such that thrifty BPs can efficiently simulate semantic incremental BPs that solve instances of $\GEN$ from the range of $g$.\footnote{More precisely, if $E$ is the set of $\BDE$ instances, and there is a size $s$ semantic incremental BP solving the set of $\GEN$ instances $g(E)$, then there is a thrifty BP solving $\BDE$ of size at most $s$.} 
Corollary \ref{c:strong_dag_lb} uses Theorem \ref{t:semantic_dag_main} for each member of a particular hard-to-pebble family of dags, whose existence was proved in \cite{PTC77}.

\subsection{Remarks on proofs by G\'{a}l, Kouck\'{y}, McKenzie} \label{s:remarksMcKenzie} 
The authors of \cite{McKenzie} obtain exponential lower bounds for \underline{syntactic} incremental BPs solving $\GEN[\m]$ in two ways. Both methods also work for a nondeterministic variant of syntactic incremental BPs. First, they use the Raz/McKenzie lower bounds for monotone circuits \cite{RazMcKenzie} to get a lower bound of $2^{n^\epsilon}$ for some $\epsilon > 0$ and sufficiently large $n$\footnote{$\epsilon$ not given explicitly.}. The first method works for a possibly-larger larger class of BPs, but the definition of that class is not simple.\footnote{See section 3.1 ``Tight Computation of GEN'' of \cite{McKenzie}.} Their second method uses a probabilistic argument\footnote{See Lemma 5.2 ``Symmetrization Lemma'' of \cite{McKenzie}.} combined with the same pebbling result that we use to get a lower bound of $2^{c n / \log n}$ for some $c > 0$ and sufficiently large $n$. 

\section{Results} \label{s:results}
\newcommandx{\hardinp}[2][1=G,2=k]{D_{#1,#2}}
\subsection{Lower Bound for Thrifty BPs} \label{s:thrifty_lb}

\begin{theorem}\label{t:thrifty_lb}
If $G$ has pebbling cost $p$ then for any $k \ge 2$ every thrifty deterministic BP solving $\BDE[G][k]$ has at least  $k^p$ states.\footnote{In section \ref{s:prelims} we specified that $G$ denotes a connected rooted DAG with at least two nodes.}
\end{theorem}

\begin{proof}
Fix $G$, $k$ and a deterministic thrifty BP $B$ that solves $\BDE[G][k]$. Let $n$ be the number of nodes in $G$ and $Q$ the states of $B$. If $u$ is a non-leaf node with $d$ children then the $u$ variables are $f_u(\av)$ for each $\av \in [k]^d$, and if $u$ is a leaf node then there is just one $u$ variable $l_u$. We sometimes say ``$f_u$ variable'' just as an in-line reminder that $u$ is a non-leaf node. 
When it is clear from the context that a state $q$ is on the computation path of an input $I$, we just say ``$q$ queries $u$'' instead of ``$q$ queries the thrifty $u$ variable of $I$''. 

We want to assign a black pebbling sequence to each input; to do this we need the following lemma.
\begin{lemma}\label{l:basic_thrifty}
   For any input $I$ and non-leaf node $u$, there is at least one state $q$ on the computation path of $I$ that queries $u$,\footnote{Recall that ``queries $u$'' means queries the thrifty $u$ variable of $I$.} and for every such $q$, for each child $v$ of $u$ there is a state on the computation path of $I$ before $q$ that queries $v$.
 \end{lemma}
 \begin{proof}
 Fix an input $I$. We prove the lemma for $I$ starting with the root, and then the children of the root, and so on. Let $v_1,\ldots,v_d$ be the children of $\root$. $I$ must visit at least one state that queries its thrifty $\root$ variable, since otherwise $B$ would make a mistake on an input $J$ that is identical to $I$ except 
\[
f_{u}^{J}(\val[v_1], \ldots, \val[v_d]) = 
\begin{cases} 
2 & \text{if } \val[\root] = 1 \\
1 & \text{otherwise}
\end{cases}
\] 
Now let $u$ be any non-leaf node and $q$ any state on the computation path of $I$ that queries $u$. Suppose the lemma does not hold for this $q$, so for some child $v$ of $u$ there is no state before $q$ that queries $v$. For every $a \not = \val[v][I]$ there is an input $I_a$ that is identical to $I$ except $\val[v][I_a] = a$. Now $I_a$ visits $q$ since $I$ and $I_a$ have the same computation path up to $q$; hence the thrifty assumption is violated. 
\end{proof}

We define the pebbling sequence for each input $I$ by following the computation path of $I$ from beginning to end, associating the $t$-th state visited by $I$ with the $t$-th pebbling configuration $C_t$, such that $C_{t+1}$ is either identical to $C_t$ or follows from $C_t$ by applying a valid pebbling move. 
Let $q_1,\ldots,q_{t^*}$ be the states on the computation path of $I$ up to the state $q_{t^*}$ immediately following the first state that queries the root; $C_{t^*}$ will be the last configuration, and the only configuration where the root is pebbled. Note that $q_1$ must query a leaf by Lemma \ref{l:basic_thrifty}. We associate $q_1$ with the empty configuration $C_1$.

Assume we have defined the configurations $C_1,\ldots,C_t$ associated with the first $t < t^*$ states, and assume $C_1,\ldots,C_t$ is a valid sequence of configurations (where adjacent identical configurations are allowed), but neither it nor any prefix of it is a complete pebbling sequence. We also maintain that for all $t' \le t$, if the node queried by $q_{t'}$ is not a leaf, then its children are pebbled in $C_{t'}$ and it is not. Let $u$ be the node queried by $q_t$. By the I.H. $u$ is not pebbled in $C_t$. We define $C_{t+1}$ by saying how to obtain it by modifying $C_{t}$:
\begin{ppe} 
  \item If $u$ is the root, then $t+1=t^*$ by the definition of $q_{t^*}$, and by the I.H. all the children of $u$ are pebbled. Put a pebble on the root and remove the pebbles from its children. This completes the definition of the pebbling sequence for $I$. 
  \item If $u$ is not the root or a leaf, then by the I.H. all the children of $u$ are pebbled. 
   For each child $v$ of $u$: if there is a state $q'$ after $q_t$ that queries some parent of $v$, and no state between $q_{t}$ and $q'$ that queries $v$, then leave the pebble on $v$, and otherwise remove it.
  \item If $u$ is not the root, then place a pebble on it iff there is a state $q'$ after $q_t$ that queries some parent of $u$ and no state between $q_{t}$ and $q'$ that queries $u$.
\end{ppe}

Let $p^I$ be the maximum number of pebbled nodes over all the configurations we just defined. So $p^I \ge p$ since $G$ has pebbling cost $p$. Let $C_t$ be the earliest configuration with $p^I$ pebbled nodes. Later we will need that $q_t$ is not an output state, so we prove that now. 
 \vspace{5pt}

\begin{indentedpar} It suffices to show $t < t^*$, since then there must be at least one state $q_{t+1}$ (possibly an output state) after $q_t$. We use the assumption that $G$ is connected and has at least two nodes, so the root has degree $d \ge 1$. In the move from $C_{t^*-1}$ to $C_{t^*}$ one pebble is added and $d$ pebbles are removed, so either $C_{t^*}$ has fewer than $p^I$ pebbled nodes (if $d > 1$) or else $C_{t^*-1}$ is an earlier configuration with $p^I$ pebbled nodes. Hence $t < t^*$. 
\end{indentedpar} \\
Define the {\bf critical state} $r^I$ for $I$ to be $q_{t}$. We refer to the nodes pebbled in $C_{t}$ as the {\bf bottleneck nodes} of $I$.  
The following fact is immediate from the pebbling sequence assignment.

\begin{fact}\label{f:basic_peb_seq}
For any input $I$, if non-root node $u$ has a pebble at a state $q$ (i.e. the configuration associated with $q$), then there is a later state $q'$ that queries some parent of $u$ and no state between (inclusive) $q$ and $q'$ that queries $u$. 
\end{fact}
Let $D$ be the set of inputs $I$ such that for every non-leaf node $u$, if $v_1,\ldots,v_d$ are the children of $u$ then $f_u^I(\vec{a}) = 1$ except possibly when $\vec{a} = \pair{\val[v_1],\ldots,\val[v_k]}$. So $|D| = k^n$. Let $R$ be the states that are critical for at least one input in $D$, and for each $r \in R$ let $D_r$ be the inputs in $D$ with critical state $r$. The remainder of the proof of Theorem \ref{t:thrifty_lb} is devoted to the proof of the next lemma. 

\begin{lemma}\label{l:thrifty_main_lemma_two}
$|D_r| \le k^{n-p}$ for every $r \in R$ 
\end{lemma}

Let us first see that the theorem follows from the lemma. Since $\{D_r\}_{r \in R}$ is a partition of $D$, by the lemma there must be at least $|D| / k^{ n-p} = k^p$ sets in the partition, i.e. the set of critical states $R$ has size at least $k^p$, which is what we wanted to show.

\newcommand{\D}{E_r}
\newcommand{\tP}{\tilde{P}}
\mksf{\nilpath}{\emptyset}

Consider a very simple cooperative two player game where Player 1 chooses $r \in R$ and an input $I$ in $D_r$ and gives $r$ to Player 2. Both players know the branching program $B$. Player 2's goal is to determine $I$ (which is the only thing Player 1 knows that Player 2 does not), which by the definition of $D$ is equivalent to determining the node values of $I$. Player 1 gets to send an advice strings to Player 2, and it is her goal to minimize the length of the advice strings. The lemma says that, for any critical state $r$ chosen by Player 1, advice strings in $[k]^{n-p}$ suffice to enable Player 2 to determine the input in $D_r$ chosen by Player 1. We refer to the individual elements from $[k]$ of an advice string as \emph{words}. 

Fix $r$ in $R$. Let $I \in D_r$ be an input chosen by Player 1, unknown to Player 2. Player 2 will use the advice, together with $r$ and the thrifty property of $B$, to follow the computation path taken by $I$ from $r$ till $I$'s output state. We will define the advice string so that each word tells Player 2 the value of a different node; when Player 2 learns (the value of) a node in this way, we say he \emph{receives} the value of ($I$ on) that node. There will be at least $p$ nodes --specifically, the bottleneck nodes of $I$-- that Player 2 will \emph{not} receive the values of, but by using the thrifty property he will learn them nonetheless; when Player 2 learns a node in this way, we say he \emph{deduces} the value of that node. 

Let $q$ be the state Player 2 is currently on, initially $q = r$. Let $u$ be the node queried by $q$. Suppose $u$ is an internal node, and let  $f_u(a_1,\ldots,a_d)$ be the variable queried by $q$ and $v_1,\ldots,v_d$ the children of $u$. Since $B$ is thrifty, $a_1,\ldots,a_d$ are the values of $I$ on $v_1,\ldots,v_d$. Hence, for each $v_i$, if Player 2 does not yet know $I(v_i)$ (meaning, he did not in some earlier state \emph{receive} or \emph{deduce} the value of $v_i$) then he deduces $I(v_i) = a_i$ now. Next, Player 2 needs to decide what edge out of $q$ to follow (Player 2 does this step for all nodes $u$, including leaf nodes). If for some $a$ he learned $I(u) = a$ at an earlier state, then he again takes the edge labeled $a$. Otherwise, we define the next unused word in the advice string to be $I(u)$, and Player 2 uses that word now. 

It is clear that for some $m \le n$, the protocol just defined will allow Player 2 to reach the output state of $I$ and learn at least $m$ node values along the way, using at most $m$ words of advice. We will argue for a stronger proposition: for some $m \le n-p$, a string of $m$ words suffices to allow Player 2 to reach the output state and learn at least $m + p$ nodes along the way. That will finish the proof of the lemma, since then we can use the remaining $(n - p) - m$ words of the advice string for $I$ for the values of the $\le n - (m + p)$ remaining nodes (ordered by some globally-fixed order on the nodes of $G$) that Player 2 has not yet learned. 
Now, by Fact \ref{f:basic_peb_seq}, for every bottleneck node $u$ of $I$, some parent of $u$ is queried at some state on the path from $r$ to the output state of $I$. Furthermore, if $u$ is ever queried on the path from $r$ till the output state, then this must happen \emph{after} some parent of $u$ is queried. Hence, for every bottleneck node $u$, Player 2 will be able to \emph{deduce} $u$ before he is forced to use a word of the advice to \emph{receive} the value of $u$. Since the nodes whose values are deduced by Player 2 are disjoint from the nodes whose values are received by Player 2, and Player 2 uses $m$ words by assumption, in total Player 2 learns at least $m + p$ node values\footnote{We do not say \emph{exactly} $m+p$ because technically, according to the given protocol, if the bottleneck configuration assigned to input $I$ has $p' > p$ pebbles (which can happen), then Player 2 learns more than $m+p$ node values.}.
\end{proof}

\subsection{Lower bound for Thrifty BPs using easy inputs} \label{s:easy_inputs}

For much of the proof of Theorem \ref{t:thrifty_lb}, we only considered the behavior of the thrifty BP on inputs from the following set:
\begin{defn}[hard inputs for thrifty programs] 
For given dag $G$ and $k \ge 2$, let $\hardinp$ be the set of $\BDE$ inputs $I$ such that for every non-leaf node $u$, if $v_1,\ldots,v_d$ are the children of $u$ then $f_u^I(\vec{a}) = 1$ except possibly when $\vec{a} = \pair{\val[v_1],\ldots,\val[v_k]}$.
\end{defn}
The sets $\hardinp$ are a small fraction of the $\BDE$ inputs, and it is not hard to see that separating the YES and NO instances of $\hardinp$ is easy for unrestricted BPs. 
The next result shows that the bound of Theorem \ref{t:thrifty_lb} holds even for thrifty BPs that are only required to be correct on inputs from $\hardinp$.

\begin{theorem}\label{t:thrifty_lb_easyinputs}
If $G$ has pebbling cost $p$ then for any $k \ge 2$ if $B$ is a thrifty deterministic BP that computes a set consistent with $\BDE[G][k]$ for the inputs $\hardinp$, then $B$ has at least  $k^p$ states. 
\end{theorem}
\begin{proof}
In the proof of Theorem \ref{t:thrifty_lb}, the only place we used the correctness of $B$ on inputs outside of $\hardinp$ is in the proof of Lemma \ref{l:basic_thrifty}. We now proof it under the weaker assumptions.  
\begin{lemma}\label{l:basic_thrifty_easyinputs}
   For any input $I \in \hardinp$ and non-leaf node $u$, there is at least one state $q$ on the computation path of $I$ that queries $u$, and for every such $q$, for each child $v$ of $u$ there is a state on the computation path of $I$ before $q$ that queries $v$.
 \end{lemma}
  Fix $I \in \hardinp$. Once again we prove the lemma for $I$ starting with the root, and then the children of the root, and so on. Let $v_1,\ldots,v_d$ be the children of $\root$. $I$ must visit at least one state that queries its thrifty $\root$ variable, since otherwise $B$ would make a mistake on an input $J \in \hardinp$ that is identical to $I$ except 
\[
f_{u}^{J}(\val[v_1], \ldots, \val[v_d]) = 
\begin{cases} 
2 & \text{if } \val[\root] = 1 \\
1 & \text{otherwise}
\end{cases}
\] 
Now let $u$ be any non-leaf node and suppose there is some state on the computation path of $I$ that queries $u$ for which the lemma does not hold. Let $q$ be the earliest such state. So for some child $v$ of $u$ there is no state before $q$ that queries $v$. For any $a \not = \val[v][I]$ there is an input $I_a \in \hardinp$ that is identical to $I$ except $\val[v][I_a] = a$ and $f_{u}^{I_a}(\val[v_1][I], \ldots, \val[v_d][I]) = 1$. Suppose there is no state $q'$ before $q$ on the computation path of $I$ that queries $u$. Then $I$ and $I_a$ have the same computation path up to $q$ and so $I_a$ visits $q$ also, which violates the thrifty assumption. Hence there is a state $q'$ before $q$ on the computation path of $I$ that queries $u$. By our choice of $q$, the lemma must hold for $q'$. This is a contradiction since the states given by the conclusion of the lemma for $q'$ satisfy the conclusion of the lemma for $q$ as well.
\end{proof}

\subsection{Lower Bound for Semantic-incremental BPs} \label{s:sem_incr_lb}

\begin{theorem}\label{t:semantic_dag_main}
If there is a rooted DAG $G$ with $n \ge 2$ nodes, indegree $2$ and pebbling cost $p$, then for any $k \ge 2$ and $\m = 3kn+n+1$ every deterministic semantic incremental BP solving $\GEN[\m]$ has at least $k^p$ states.
\end{theorem}
\begin{proof}

Let $G,n,p,k,\m$ be as in the statement of the theorem; these are fixed throughout the proof. 
The bulk of this argument is a reduction from $\BDE[G][k]$ to $\GEN[\m]$; we map each instance $I$ of $\BDE[G][k]$ to an instance $T^I$ of $\GEN[\m]$ such that $I$ is a YES instance iff $T^I$ is. 
Let $E$ be the set of inputs for $\BDE[G][k]$. We will show that if there is a semantic-incremental $\m$-way BP of size $s$ that computes $\GEN[\m]$ correctly on the inputs $\{T^I\}_{I \in E}$, then there is a thrifty $k$-way BP of size at most $s$ that computes $\BDE[G][k]$. Then from Theorem \ref{t:thrifty_lb} we get $s \ge k^p$. 

Fix an order on the nodes of $G$. We will not differentiate between a node $u$ and its index in $[n]$ given by this order. Let $\root \in [n]$ be the (index of) the root of $G$.
We divide the elements $\{n+2,\ldots,m\}$ into two parts, one of size $nk$ and the other of size $\le 2nk$,\footnote{$2n$ is a bound on the number of edges since non-leaf nodes have indegree at most 2.} and refer to them by the following mnemonic names:
\begin{ppi}
 \item \nhasv{u}{a} for each node $u$ and $a \in [k]$. $T^I$ generates this element iff $\val[u] = a$.
 \item \ehasv{vu}{a} for each arc $vu$ and $a \in [k]$. $T^I$ generates this element iff $\val[v] = a$.\footnote{These elements may seem redundant, given the elements $\nhasv{u}{a}$. However, that is correct only if $G$ has the property that no two nodes $u_1,u_2$ have the same children.}
\end{ppi}
Any way of assigning those $\le 3nk$ names to distinct elements of $\{ n+2,\ldots,m \}$ will suffice, except we require that $\nhasv{\root}{1}$ gets assigned to $m$, since then we will have that $T^I$ is a YES instance of $\GEN[m]$ iff $I$ is a YES instance of $\BDE$.
Elements $\{1,\ldots,n+1\}$ will be generated by every $T^I$; their purpose is technical. 
Now we give the reduction. Fix an instance $I$ of $\BDE[G][k]$. 
First we make $T^I$ generate each of the elements $\{2,\ldots,n+1\}$. 
For each $u \in [n]$:
\begin{equation} 
\begin{split}
T^I(\initial, u) & := \indexofnode{u} 
\end{split}
\label{e:defTechnical} \end{equation}
In a similar way, we make $T^I$ generate the elements \nhasv{w}{l_{w}^I}  for each leaf $w$. Fix an order $w_1,\ldots,w_l$ on the leaf nodes. For each $t \in [l-1]$ and $a \in [k]$: 
\begin{equation} 
\begin{split}
 T^I(\initial,\indexofnode{n}) & := \nhasv{w_1}{l_{w_1}^I}  \\
 T^I(\initial, \nhasv{w_t}{a}) & := \nhasv{w_{t+1}}{l_{w_{t+1}}^I} 
 \end{split}
\label{e:defLeaves} \end{equation}
For every non-leaf node $u \in [n]$ and $a,b_1,b_2 \in [k]$, if $v_1$ and $v_2$ are the left and right children of $u$ then we add the following definitions. Equations (\ref{e:defInterp}) simply propagate the value of a node to its out-arcs.  Let $b \eqdef f_u^I(b_1,b_2)$. Equation (\ref{e:defComp}) expresses: If $I$ gives the left in-arc of $u$ value $b_1$ and the right in-arc of $u$ value $b_2$, then $I$ gives $u$ the value $b = f_u^I(b_1,b_2)$. 
\begin{equation} \label{e:defInterp} 
\begin{split}
 T^I(\indexofnode{u}, \nhasv{v_1}{a}) & := \ehasv{v_1 u}{a} \\
  T^I(\indexofnode{u}, \nhasv{v_2}{a}) & := \ehasv{v_2 u}{a} 
\end{split}
\end{equation}
\begin{equation} 
 T^I(\ehasv{v_1 u}{b_1}, \ehasv{v_2 u}{b_2})  := \nhasv{u}{b} 
\label{e:defComp} \end{equation}
Let us call a variable $(x,y)$ \textbf{used} if $T^I(x,y)$ is defined at this point, and \textbf{unused} otherwise. 
Examining the left sides of equations (\ref{e:defTechnical})--(\ref{e:defComp}), it is clear that the set of used variables depends only on $G$ and $k$ (not on $I$). For every unused variable $(x,y)$ define 
\begin{equation} 
T^I(x,y) := \initial
\label{e:defUnused}  \end{equation}
That completes the definition of $T^I$. It is straightforward to show that $I$ is a YES instance of $\BDE$ iff $T^I$ is a YES instance of $\GEN[m]$.

Now we show how to convert, without increasing the size, a semantic-incremental $\m$-way BP $B$ that computes $\GEN[\m]$ correctly on the inputs $\{T^I\}_{I \in E}$, into a thrifty $k$-way BP that computes $\BDE[G][k]$. 
From the above definition of the inputs $\{T^I\}_{I \in E}$, an $\GEN[\m]$ variable $(x,y)$ is a used variable iff it has exactly one of the following four forms, where \ref{i:nodevtype} corresponds to Equation (\ref{e:defTechnical}), and \ref{i:leafvtype} corresponds to Equation (\ref{e:defLeaves}), etc:
\begin{ppe} 
\renewcommand{\theenumi}{\roman{enumi}.}
\renewcommand{\labelenumi}{\theenumi}
 \item $( \initial, u )$ for some $u \in [n]$    \label{i:nodevtype}
 \item $(\initial, \nhasv{w_t}{a})$ for some leaf node $w_t$ and $a \in [k]$ \label{i:leafvtype}
 \item $(\indexofnode{u}, \nhasv{v}{b} )$ for some $b \in [k]$ and some non-leaf node $u \in [n]$ with child $v$ \label{i:techvtype}
 \item $(\ehasv{v_1u}{b_1}, \ehasv{v_2 u}{b_2} )$ for some non-leaf node $u$ with children $v_1,v_2$ and $b_1,b_2 \in [k]$ \label{i:intvtype}
\end{ppe}
We say $(x,y)$ is a type \ref{i:nodevtype} variable iff it has the form of \ref{i:nodevtype} above, and type \ref{i:leafvtype}, \ref{i:techvtype}, \ref{i:intvtype} variables are defined analogously.
Only variables of type \ref{i:leafvtype} and \ref{i:intvtype} will be translated to $\BDE$ variables. The remaining used variables, of types \ref{i:nodevtype} and \ref{i:techvtype}, are ``\textbf{dummy}'' variables, in the sense that the right sides of the corresponding defining equations (\ref{e:defTechnical}) and (\ref{e:defInterp}) do not depend on $I$. So for any state $q$ in $B$ that is labeled with a dummy variable, there is at most one edge out of $q$ that any input $T^I$ can take; as a consequence these states will eventually be deleted. For the same reason, states labeled with unused variables will also be deleted.
 
Recall the ordering on the leaf nodes $w_1,\ldots,w_{l}$ that we fixed earlier. Let $q$ be a state of $B$ that queries a variable $(x,y)$. If $(x,y)$ is an unused variable then delete all edges out of $q$ except the one labeled $1$. Otherwise $(x,y)$ is one of the variable types \ref{i:nodevtype}--\ref{i:intvtype}, and should be handled as follows: 
\begin{romleftppe}
\item $(x,y) = (\initial,u)$ for some $u \in [n]$. Delete every edge out of $q$ except the one labeled $u+1$. 
\item If $(x,y) = (\initial,\indexofnode{n})$, then delete every edge out of $q$ except for the $k$ edges labeled $\nhasv{w_1}{a}$ for $a \in [k]$. \\
Otherwise $(x,y) = (\initial, \nhasv{w_t}{a})$ for some $t \in [l]$ and $a \in [k]$. Delete every edge out of $q$ except for the $k$ edges labeled $\nhasv{w_{t+1}}{b}$ for $b \in [k]$.
\item  $(x,y) = (\indexofnode{u}, \nhasv{v}{b})$ for some $b \in [k]$ and nodes $u,v$ such that $v$ is a child of $u$. Delete every edge out of $q$ except the one edge labeled $\ehasv{vu}{b}$.
\item $(x,y) = (\ehasv{v_1 u}{b_1}, \ehasv{v_2 u}{b_2})$ for some non-leaf node $u$ with children $v_1,v_2$ and $b_1,b_2 \in [k]$. Delete every edge out of $q$ except the $k$ edges labeled $\nhasv{u}{a}$ for $a \in [k]$.
\end{romleftppe}

At this point, a non-output state $q$ has outdegree zero or one iff it is labeled with an unused or dummy variable. For any non-output state $q$ with outdegree zero, change $q$ to a reject state. Now consider $q$ with outdegree one. Note its out-edge must be labeled with an element of the same form as the right side of one of Equations (\ref{e:defTechnical}), (\ref{e:defInterp}) or (\ref{e:defUnused}). Let  $q'$ be the unique state that $q$ transitions to. For every state $q''$ with an out-edge $e$ to $q$, move the target node of $e$ from $q$ to $q'$. After all the edges into $q$ have been moved, delete $q$. 
The last step is to rename the variable labels on the remaining states. Note that only types \ref{i:leafvtype}, and \ref{i:intvtype} remain. 
Rename them as follows:
\[ (\initial, \indexofnode{n}) \mapsto l_{w_1} \quad \ \  (1, \nhasv{w_t}{a}) \mapsto l_{w_{t+1}} \] 
\[ (\ehasv{v_1 u}{b_1},\ehasv{v_2 u}{b_2}) \mapsto f_u(b_1,b_2) \]
Let $B'$ be the resulting branching program. Since we have obtained $B'$ only by renaming variables and deleting edges and states, clearly its size is no greater than that of $B$. We need that $B$ accepts $T^I$ iff $B'$ accepts $I$ for every $I \in E$. It is clear that ``bypassing'' and then deleting the states labeled by unused and dummy variables, in the way done above, has no effect on any input $T^I$. Also, one can check that for every edge $e$ we removed before deleting the dummy states, none of the inputs $T^I$ can take edge $e$. Finally, $B'$ is thrifty precisely because $B$ is semantic-incremental.
\end{proof}

\begin{cor} \label{c:strong_dag_lb}
There is a constant $c>0$ such that for infinitely-many $\m$ every deterministic semantic-incremental BP solving  $\GEN[\m]$ has at least $2^{c\, m / \log m}$ states.
\end{cor}
\begin{proof}
 We will need a family of rooted DAGs that are much harder to pebble than the complete binary trees. \cite{PTC77} provides such a family:
\begin{quote}
\emph{There is a family of rooted DAGs $\{G_t\}_{t\ge1}$ with $\Theta(t)$ nodes and indegree two whose black pebbling cost is $\Omega(t / \log t)$.}
\end{quote}
Let $n_t$ be the number of nodes in $G_t$\footnote{$n_t \ge 2$ for all $t \ge 1$, but regardless we could always just start at some $t_0 > 1$ to ensure that, without changing the result.}, and let $p(n_t)$ be the pebbling cost of $G_t$ as a function of $n_t$. So $p(n_t) = \Omega(n_t / \log n_t)$. We use Theorem \ref{t:semantic_dag_main} on each of the $G_t$ with $k=2$. For $\m_t := 3kn_t + n_t + 1 = 7n_t + 1 \le 8 n_t$, for every $t$ we get lower bounds of $2^{p(n_t)} \ge 2^{p(\m_t/8)}$ for semantic incremental BPs solving $\GEN[m_t]$. 
This suffices since $p(\m_t/8) \ge c\, \m_t / \log \m_t$ for some constant $c>0$ and all $t$.
\end{proof}

\section{Open problem} 
\emph{Nondeterministic thrifty  BPs solving $\BT$} were defined and studied \cite{FSTTCS} and \cite{myMasters}. The definition of \emph{nondeterministic thrifty BPs solving $\BDE[G]$} is the natural extension of that definition from binary trees to arbitrary dags. \cite{McKenzie} defined and studied a similar model, \emph{nondeterministic semantic incremental BPs solving GEN}. The open problem is to prove superpolynomial lower bounds for one or both models. An interested reader should consult \cite{FSTTCS} and \cite{McKenzie} for upper bound results and more. It should be possible to adapt the main construction in the proof of Theorem \ref{t:thrifty_lb} (converting deterministic semantic incremental BPs into deterministic thrifty BPs) to show that lower bounds for nondeterministic semantic incremental BPs follow from lower bounds for nondeterministic thrifty BPs.\footnote{But an easy proof of the converse appears unlikely, due to a subtlety in the definition of nondeterministic semantic incremental BPs. An interested reader should contact the author of this paper.}

\section{Acknowledgements}
Thanks to Steve Cook and Pierre McKenzie for many helpful comments on drafts of this paper.  

\bibliographystyle{alpha}
\bibliography{semantic_incremental_dags}

\end{document}